\documentclass[a4paper,11pt,english]{llncs}

\usepackage[english]{babel}

\usepackage{latexsym,amsfonts,amsmath,amssymb,amscd,stmaryrd,mathrsfs}

\newtheorem{innercustomthm}{Theorem}
\newenvironment{customthm}[1]
  {\renewcommand\theinnercustomthm{#1}\innercustomthm}
  {\endinnercustomthm}

\usepackage{graphicx,graphics}
\usepackage{adjustbox}
\usepackage{float}
\usepackage{color}
\usepackage[table]{xcolor}

\usepackage{tikz}
\usetikzlibrary{patterns,positioning,arrows,decorations.pathmorphing}

\definecolor{rouge}{RGB}{255,77,77}
\definecolor{vert}{RGB}{0,178,102}
\definecolor{jaune}{RGB}{255,255,0}
\definecolor{violet}{RGB}{208,32,144}
\definecolor{orange}{RGB}{255,140,0}
\definecolor{bleu}{RGB}{0,0,205}

\renewcommand{\L}{\mathbb{L}}

\newcommand{\MT}{\mathcal{M}}
\newcommand{\N}{\mathbb{N}}

\renewcommand{\H}{\mathbb{H}}

\newcommand{\Z}{\mathbb{Z}}

\newcommand{\U}{\mathbb{U}}

\newcommand{\T}{\mathbf{T}}

\newcommand{\define}{\emph}

\newcommand{\bnoir}[1]{\vbox to 3pt{\hbox{
\begin{tikzpicture}[scale=0.6*#1]
\draw [fill=black] (5,1) rectangle (5.5,1.5);
\end{tikzpicture}
}}}

\newcommand{\bnoirrouge}[1]{\vbox to 3pt{\hbox{
\begin{tikzpicture}[scale=0.6*#1]
\draw[color=rouge,fill=rouge] (5,1) rectangle (5.5,1.5);
\end{tikzpicture}
}}}

\newcommand{\bnoirbleu}[1]{\vbox to 3pt{\hbox{
\begin{tikzpicture}[scale=0.6*#1]
\draw[color=bleu,fill=bleu] (5,1) rectangle (5.5,1.5);
\end{tikzpicture}
}}}

\newcommand{\bnoirvert}[1]{\vbox to 3pt{\hbox{
\begin{tikzpicture}[scale=0.6*#1]
\draw[color=vert,fill=vert] (5,1) rectangle (5.5,1.5);
\end{tikzpicture}
}}}

\newcommand{\bblanc}[1]{\vbox to 3pt{\hbox{
\begin{tikzpicture}[scale=0.6*#1]
\draw (5,1) rectangle (5.5,1.5);
\end{tikzpicture}
}}}

\newcommand{\bblancrouge}[1]{\vbox to 3pt{\hbox{
\begin{tikzpicture}[scale=0.6*#1]
\draw[color=rouge] (5,1) rectangle (5.5,1.5);
\end{tikzpicture}
}}}

\newcommand{\bblancbleu}[1]{\vbox to 3pt{\hbox{
\begin{tikzpicture}[scale=0.6*#1]
\draw[color=bleu] (5,1) rectangle (5.5,1.5);
\end{tikzpicture}
}}}

\newcommand{\bblancvert}[1]{\vbox to 3pt{\hbox{
\begin{tikzpicture}[scale=0.6*#1]
\draw[color=vert] (5,1) rectangle (5.5,1.5);
\end{tikzpicture}
}}}

\newcommand{\hypertiledyadic}[5]{
\begin{scope}[shift={(#1,#2)},scale=#3]
\draw (0,0) -- (0,2) -- (0,2) to [controls=+(45:1.5) and +(135:1.5)] (4,2) -- (4,0) to [controls=+(135:0.75) and +(45:0.75)] (2,0) to [controls=+(135:0.75) and +(45:0.75)] (0,0) -- cycle ;
\draw (1,1.5) node{#4};
\draw (3,1.1) node{#5};
\end{scope}
}

\newcommand{\fleche}[5]{
\begin{scope}[color=#4,shift={(#1,#2)},scale=#3]
\draw[#5,->] (1,0.8) -> (1.5,-0.3); 
\end{scope}
}

\newcommand{\flechebis}[5]{
\begin{scope}[color=#4,shift={(#1,#2)},scale=#3]
\draw[#5,->] (3,0.5) -> (3.5,-0.3); 
\end{scope}
}

\title{Row-constrained effective sets of colourings in the $2$-fold horocyclic tessellations of $\H^2$ are sofic.}
\date{}
\author{Nathalie Aubrun\inst{1} and Mathieu Sablik\inst{2}}

\institute{LIP, ENS de Lyon -- CNRS -- INRIA -- UCBL, Universit\'e de Lyon
\and
I2M, Aix-Marseille Universit\'e\\
\email{nathalie.aubrun@ens-lyon.fr}, \email{sablik@latp.univ-mrs.fr}
}

\begin{document}

\maketitle
\begin{abstract}
In this article we prove that, restricted to the row-constrained case, effective sets of colourings in the $2$-fold horocyclic tessellations of the hyperbolic plane $\H^2$ are sofic.
\end{abstract}

\section*{Introduction}

Multidimensional subshifts of finite type (SFT) and sofic subshifts are closed and shift-invariant subsets of colourings of $\Z^d$ for $d\geq2$ given by local rules, and enjoy strong computational properties. For instance, it is not possible to decide whether such a subshift is empty or not~\cite{Berger1966}. A clever result by Hochman~\cite{hochman2007drp}, then improved independently in~\cite{AubrunSablik2013} and~\cite{DBLP:conf/birthday/DurandRS10}, states that with an increase by $1$ of the dimension, effective subshifts are very close to sofic subshifts. 
Symbolic dynamics can be defined on structures more general than $\Z^d$, for instance finitely presented groups, and a natural question is to determine whether results similar to Hochman's can be proved in this case.


Our intuition is that we can obtain an even stronger result than Hochman's on the hyperbolic plane. Two facts strengthen our intuition. First it is possible to encode Turing machine computations with local rules on the hyperbolic plane~\cite{Robinson1978}. Second the counting argument (see~\cite[p.~14]{Vanier2012}) used to prove the non soficness of the mirror subshift in $\Z^2$ -- one example of effective subshift that can be proved not to be sofic -- cannot be applied on the hyperbolic plane, for non-amenability reasons. This leads us to formulate the following conjecture, that basically means that the dimension increase is no longer needed to get a result similar to Hochman's in the hyperbolic plane.

\begin{conjecture}
\label{conjecture}
Effective sets of colourings are sofic on the hyperbolic plane.
\end{conjecture}

Unfortunately we are for now unable to prove this conjecture, but we present here a preliminary result that will hopefully be a first step in proving Conjecture~\ref{conjecture}. The idea is to simplify the problem by enforcing the border of the configurations. In the Poincar\'e disk model, this corresponds to choosing a small disk tangent to the border of the whole Poincar\'e disk, and forcing colours inside that small disk. In the Poincar\'e upper half-plane model, this means that one horizontal line and everything above it is fixed. Note that the sets of configurations we now consider are no longer subshifts, since they do not satisfy any shift-invariance property. Nevertheless this approach makes sense for at least two reasons. First is corresponds to the intuitive vision of the way somebody would try to tile a surface with a set of tiles: the person would start his tiling on the border of the surface. The second reason is historical: before the proof of the undecidability of the domino problem on $\Z^2$ 
by Berger~\cite{Berger1966} (deciding whether an SFT is empty or not) Wang first proved that the row-constrained problem, where a single tile is forced to appear, is undecidable~\cite{Wang1961}. Also in the case of the hyperbolic plane, before Kari's~\cite{DBLP:conf/mcu/Kari07} and Margenstern's~\cite{Margenstern2008} proofs of the undecidability of the unconstrained domino problem, Robinson remarked that the origin fixed domino problem was also undecidable~\cite{Robinson1978}. Note that here we consider something a little bit more general than fixing a single tile, since we fix an entire line.

The paper is organized as follows. In Section~\ref{section:tilings_hyperbolic_plane} we the present the $2$-fold horocyclic tessellations of of the Poincar\'e upper half-plane model, then define colourings on this structure and explain how to encode Turing machine computations inside such objects. Section~\ref{section:dyadic_encoding} is devoted to the dyadic encoding, a basic transformation on colourings that exploits the hyperbolic structure of our model. Finally in Section~\ref{section:result} we prove the main result.

\section{Sets of colourings of $2$-fold horocyclic tessellations of $\H^2$}
\label{section:tilings_hyperbolic_plane}

In this section we present a formalism to define tilings on one particular family of tessellations of the hyperbolic plane $\H^2$: the $2$-fold horocyclic tessellations. We then consider generalized tilings, called sets of colourings, where local rules that define allowed configurations are not necessarily finite in number.

\subsection{$2$-fold horocyclic tessellations of $\H^2$}

In Figure~\ref{figure:semi_group_H2}, we show one \emph{$2$-fold horocyclic tessellation of $\H^2$}, depicted in the upper-half plane model of the hyperbolic plane. The tiles are arranged hierarchically, each sitting above two other tiles. Basically, one row fixes all rows below it, and there are two choices (right and left) for the row immediately above it. There are consequently uncountably many tessellations of $\H^2$ with these tiles, but in the sequel we will work with only one of them. Note that if the sequence of choices is eventually constant (i.e. always right or always left) then the tessellation contains a vertical fracture line, that separates the tessellation into two symmetrical parts. For more convenience, we will allow ourselves to locate tiles by a finite word on the alphabet $\{\alpha,\alpha^{-1},\beta,\beta^{-1} \}$. To do so, we choose to represent the tile that covers the origin point of $\H^2$ by the empty word $\varepsilon$. Then if a tile is represented by the word $g$, its bottom 
left neighbour (resp. bottom right neighbour) will be represented by $g\cdot\alpha$ (resp. $g\cdot\alpha\cdot\beta$). This rule allows to represent every tile below the tile $\varepsilon$ by a word, the rest of the tessellation is obtained by using the trivial rules $\alpha\cdot\alpha^{-1}=\beta\cdot\beta^{-1}=\varepsilon$. In order to make these 
representations consistent we also need to add the relation $\alpha\cdot\beta^2=\beta\cdot\alpha$. Note that to one tile correspond infinitely many finite words, but not all finite words correspond to a tile in the tessellation -- for instance the word $\alpha^{-1}$ does not correspond to any tile for the tessellation with choice of $\varepsilon$ presented in Figure~\ref{figure:semi_group_H2}. The set of non-valid words of course depends on the choice for $\varepsilon$, but it is always recursive.

\vspace{-0.5cm}
\begin{figure}[!h]
\centering
\begin{tikzpicture}[scale=0.5]

\begin{scope}[shift={(0,2)},scale=2]
\draw (0,0) -- (0,2) -- (0,2) to [controls=+(45:1.5) and +(135:1.5)] (4,2) -- (4,0) to [controls=+(135:0.75) and +(45:0.75)] (2,0) to [controls=+(135:0.75) and +(45:0.75)] (0,0) -- cycle ;
\draw (2,0.75) node[above]{\large$\beta^{-3}\alpha^{-1}$};
\begin{scope}[shift={(4,0)}]
\draw (0,0) -- (0,2) -- (0,2) to [controls=+(45:1.5) and +(135:1.5)] (4,2) -- (4,0) to [controls=+(135:0.75) and +(45:0.75)] (2,0) to [controls=+(135:0.75) and +(45:0.75)] (0,0) -- cycle ;
\draw (2,0.75) node[above]{\large$\beta^{-1}\alpha^{-1}$};
\end{scope}
\begin{scope}[shift={(8,0)}]
\draw (0,0) -- (0,2) -- (0,2) to [controls=+(45:1.5) and +(135:1.5)] (4,2) -- (4,0) to [controls=+(135:0.75) and +(45:0.75)] (2,0) to [controls=+(135:0.75) and +(45:0.75)] (0,0) -- cycle ;
\draw (2,0.75) node[above]{\large$\beta\alpha^{-1}$};
\end{scope}
\end{scope}

\draw (0,0) -- (0,2) -- (0,2) to [controls=+(45:1.5) and +(135:1.5)] (4,2) -- (4,0) to [controls=+(135:0.75) and +(45:0.75)] (2,0) to [controls=+(135:0.75) and +(45:0.75)] (0,0) -- cycle ;
\draw (2,0.75) node[above]{\large{$\beta^{-3}$}};
\begin{scope}[shift={(4,0)}]
\draw (0,0) -- (0,2) -- (0,2) to [controls=+(45:1.5) and +(135:1.5)] (4,2) -- (4,0) to [controls=+(135:0.75) and +(45:0.75)] (2,0) to [controls=+(135:0.75) and +(45:0.75)] (0,0) -- cycle ;
\draw (2,0.75) node[above]{\large{$\beta^{-2}$}};
\end{scope}
\begin{scope}[shift={(8,0)}]
\draw (0,0) -- (0,2) -- (0,2) to [controls=+(45:1.5) and +(135:1.5)] (4,2) -- (4,0) to [controls=+(135:0.75) and +(45:0.75)] (2,0) to [controls=+(135:0.75) and +(45:0.75)] (0,0) -- cycle ;
\draw (2,0.75) node[above]{\large{$\beta^{-1}$}};
\end{scope}
\begin{scope}[shift={(12,0)}]
\draw (0,0) -- (0,2) -- (0,2) to [controls=+(45:1.5) and +(135:1.5)] (4,2) -- (4,0) to [controls=+(135:0.75) and +(45:0.75)] (2,0) to [controls=+(135:0.75) and +(45:0.75)] (0,0) -- cycle ;
\draw (2,0.75) node[above]{\large{$\varepsilon$}};
\end{scope}
\begin{scope}[shift={(16,0)}]
\draw (0,0) -- (0,2) -- (0,2) to [controls=+(45:1.5) and +(135:1.5)] (4,2) -- (4,0) to [controls=+(135:0.75) and +(45:0.75)] (2,0) to [controls=+(135:0.75) and +(45:0.75)] (0,0) -- cycle ;
\draw (2,0.75) node[above]{\large{$\beta$}};
\end{scope}
\begin{scope}[shift={(20,0)}]
\draw (0,0) -- (0,2) -- (0,2) to [controls=+(45:1.5) and +(135:1.5)] (4,2) -- (4,0) to [controls=+(135:0.75) and +(45:0.75)] (2,0) to [controls=+(135:0.75) and +(45:0.75)] (0,0) -- cycle ;
\draw (2,0.75) node[above]{\large{$\beta^2$}};
\end{scope}

\begin{scope}[shift={(0,-1)},scale=0.5]
\draw (0,0) -- (0,2) -- (0,2) to [controls=+(45:1.5) and +(135:1.5)] (4,2) -- (4,0) to [controls=+(135:0.75) and +(45:0.75)] (2,0) to [controls=+(135:0.75) and +(45:0.75)] (0,0) -- cycle ;
\draw (2,1.5) node{$\alpha\beta^{-6}$};
\begin{scope}[shift={(4,0)}]
\draw (0,0) -- (0,2) -- (0,2) to [controls=+(45:1.5) and +(135:1.5)] (4,2) -- (4,0) to [controls=+(135:0.75) and +(45:0.75)] (2,0) to [controls=+(135:0.75) and +(45:0.75)] (0,0) -- cycle ;
\draw (2,1.5) node{$\alpha\beta^{-5}$};
\end{scope}
\begin{scope}[shift={(8,0)}]
\draw (0,0) -- (0,2) -- (0,2) to [controls=+(45:1.5) and +(135:1.5)] (4,2) -- (4,0) to [controls=+(135:0.75) and +(45:0.75)] (2,0) to [controls=+(135:0.75) and +(45:0.75)] (0,0) -- cycle ;
\draw (2,1.5) node{$\alpha\beta^{-4}$};
\end{scope}
\begin{scope}[shift={(12,0)}]
\draw (0,0) -- (0,2) -- (0,2) to [controls=+(45:1.5) and +(135:1.5)] (4,2) -- (4,0) to [controls=+(135:0.75) and +(45:0.75)] (2,0) to [controls=+(135:0.75) and +(45:0.75)] (0,0) -- cycle ;
\draw (2,1.5) node{$\alpha\beta^{-3}$};
\end{scope}
\begin{scope}[shift={(16,0)}]
\draw (0,0) -- (0,2) -- (0,2) to [controls=+(45:1.5) and +(135:1.5)] (4,2) -- (4,0) to [controls=+(135:0.75) and +(45:0.75)] (2,0) to [controls=+(135:0.75) and +(45:0.75)] (0,0) -- cycle ;
\draw (2,1.5) node{$\alpha\beta^{-2}$};
\end{scope}
\begin{scope}[shift={(20,0)}]
\draw (0,0) -- (0,2) -- (0,2) to [controls=+(45:1.5) and +(135:1.5)] (4,2) -- (4,0) to [controls=+(135:0.75) and +(45:0.75)] (2,0) to [controls=+(135:0.75) and +(45:0.75)] (0,0) -- cycle ;
\draw (2,1.5) node{$\alpha\beta^{-1}$};
\end{scope}
\begin{scope}[shift={(24,0)}]
\draw (0,0) -- (0,2) -- (0,2) to [controls=+(45:1.5) and +(135:1.5)] (4,2) -- (4,0) to [controls=+(135:0.75) and +(45:0.75)] (2,0) to [controls=+(135:0.75) and +(45:0.75)] (0,0) -- cycle ;
\draw (2,1.5) node{$\alpha$};
\end{scope}
\begin{scope}[shift={(28,0)}]
\draw (0,0) -- (0,2) -- (0,2) to [controls=+(45:1.5) and +(135:1.5)] (4,2) -- (4,0) to [controls=+(135:0.75) and +(45:0.75)] (2,0) to [controls=+(135:0.75) and +(45:0.75)] (0,0) -- cycle ;
\draw (2,1.5) node{$\alpha\beta$};
\end{scope}
\begin{scope}[shift={(32,0)}]
\draw (0,0) -- (0,2) -- (0,2) to [controls=+(45:1.5) and +(135:1.5)] (4,2) -- (4,0) to [controls=+(135:0.75) and +(45:0.75)] (2,0) to [controls=+(135:0.75) and +(45:0.75)] (0,0) -- cycle ;
\draw (2,1.5) node{$\alpha\beta^2$};
\end{scope}
\begin{scope}[shift={(36,0)}]
\draw (0,0) -- (0,2) -- (0,2) to [controls=+(45:1.5) and +(135:1.5)] (4,2) -- (4,0) to [controls=+(135:0.75) and +(45:0.75)] (2,0) to [controls=+(135:0.75) and +(45:0.75)] (0,0) -- cycle ;
\draw (2,1.5) node{$\alpha\beta^3$};
\end{scope}
\begin{scope}[shift={(40,0)}]
\draw (0,0) -- (0,2) -- (0,2) to [controls=+(45:1.5) and +(135:1.5)] (4,2) -- (4,0) to [controls=+(135:0.75) and +(45:0.75)] (2,0) to [controls=+(135:0.75) and +(45:0.75)] (0,0) -- cycle ;
\draw (2,1.5) node{$\alpha\beta^4$};
\end{scope}
\begin{scope}[shift={(44,0)}]
\draw (0,0) -- (0,2) -- (0,2) to [controls=+(45:1.5) and +(135:1.5)] (4,2) -- (4,0) to [controls=+(135:0.75) and +(45:0.75)] (2,0) to [controls=+(135:0.75) and +(45:0.75)] (0,0) -- cycle ;
\draw (2,1.5) node{$\alpha\beta^5$};
\end{scope}
\end{scope}

\end{tikzpicture}

\caption{One $2$-fold horocyclic tessellation of $\H^2$, and names given to different tiles.}
\label{figure:semi_group_H2}
\end{figure}
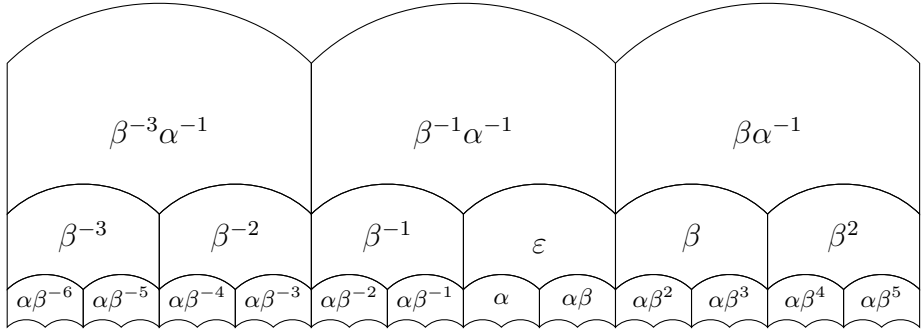
\vspace{-0.5cm}

\subsection{Sets of colourings on $\H^2$}

Suppose the origin $\varepsilon$ is fixed, and let $g\in\{\alpha,\alpha^{-1},\beta,\beta^{-1}\}^*$ be a position in $\H^2$. We denote by $\U_n$ the \emph{support of size $n$}, defined as
$$\U_n=\left\{ \alpha^p\cdot\beta^q :0\leq p\leq n-1,0\leq q\leq 2^p-1\right\},$$
and by $\L_n$ \emph{linear support of size $n$}, defined as
$$\L_n=\left\{ \alpha^{n+1}\cdot\beta^q :0\leq q\leq 2^{n+1}-1 \right\}.$$

Let $A$ be a finite alphabet. A \emph{configuration} is a colouring of the $2$-fold horocyclic tessellation of $\H^2$ with colours chosen in $A$. By abuse of notation, we denote the set of configurations by $A^{\H^2}$. A \define{pattern of size $n$} is a finite configuration $p\in A^{\U_n}$, and $\U_n$ is thus called the \define{support of $p$}. A \define{linear pattern of size $n$} is a finite configuration $p\in A^{\L_n}$.

\begin{figure}[!h]
\centering
\begin{tikzpicture}[scale=0.3]

\draw (0,0) -- (0,2) -- (0,2) to [controls=+(45:1.5) and +(135:1.5)] (4,2) -- (4,0) to [controls=+(135:0.75) and +(45:0.75)] (2,0) to [controls=+(135:0.75) and +(45:0.75)] (0,0) -- cycle ;

\begin{scope}[shift={(6,1)}]
\draw (0,0) -- (0,2) -- (0,2) to [controls=+(45:1.5) and +(135:1.5)] (4,2) -- (4,0) to [controls=+(135:0.75) and +(45:0.75)] (2,0) to [controls=+(135:0.75) and +(45:0.75)] (0,0) -- cycle ;
\begin{scope}[shift={(0,-1)},scale=0.5]
\draw (0,0) -- (0,2) -- (0,2) to [controls=+(45:1.5) and +(135:1.5)] (4,2) -- (4,0) to [controls=+(135:0.75) and +(45:0.75)] (2,0) to [controls=+(135:0.75) and +(45:0.75)] (0,0) -- cycle ;
\begin{scope}[shift={(4,0)}]
\draw (0,0) -- (0,2) -- (0,2) to [controls=+(45:1.5) and +(135:1.5)] (4,2) -- (4,0) to [controls=+(135:0.75) and +(45:0.75)] (2,0) to [controls=+(135:0.75) and +(45:0.75)] (0,0) -- cycle ;
\end{scope}
\end{scope}
\end{scope}

\begin{scope}[shift={(12,1.5)}]
\draw (0,0) -- (0,2) -- (0,2) to [controls=+(45:1.5) and +(135:1.5)] (4,2) -- (4,0) to [controls=+(135:0.75) and +(45:0.75)] (2,0) to [controls=+(135:0.75) and +(45:0.75)] (0,0) -- cycle ;
\begin{scope}[shift={(0,-1)},scale=0.5]
\draw (0,0) -- (0,2) -- (0,2) to [controls=+(45:1.5) and +(135:1.5)] (4,2) -- (4,0) to [controls=+(135:0.75) and +(45:0.75)] (2,0) to [controls=+(135:0.75) and +(45:0.75)] (0,0) -- cycle ;
\begin{scope}[shift={(4,0)}]
\draw (0,0) -- (0,2) -- (0,2) to [controls=+(45:1.5) and +(135:1.5)] (4,2) -- (4,0) to [controls=+(135:0.75) and +(45:0.75)] (2,0) to [controls=+(135:0.75) and +(45:0.75)] (0,0) -- cycle ;
\end{scope}
\end{scope}
\begin{scope}[shift={(0,-1.5)},scale=0.25]
\draw (0,0) -- (0,2) -- (0,2) to [controls=+(45:1.5) and +(135:1.5)] (4,2) -- (4,0) to [controls=+(135:0.75) and +(45:0.75)] (2,0) to [controls=+(135:0.75) and +(45:0.75)] (0,0) -- cycle ;
\begin{scope}[shift={(4,0)}]
\draw (0,0) -- (0,2) -- (0,2) to [controls=+(45:1.5) and +(135:1.5)] (4,2) -- (4,0) to [controls=+(135:0.75) and +(45:0.75)] (2,0) to [controls=+(135:0.75) and +(45:0.75)] (0,0) -- cycle ;
\end{scope}
\begin{scope}[shift={(8,0)}]
\draw (0,0) -- (0,2) -- (0,2) to [controls=+(45:1.5) and +(135:1.5)] (4,2) -- (4,0) to [controls=+(135:0.75) and +(45:0.75)] (2,0) to [controls=+(135:0.75) and +(45:0.75)] (0,0) -- cycle ;
\end{scope}
\begin{scope}[shift={(12,0)}]
\draw (0,0) -- (0,2) -- (0,2) to [controls=+(45:1.5) and +(135:1.5)] (4,2) -- (4,0) to [controls=+(135:0.75) and +(45:0.75)] (2,0) to [controls=+(135:0.75) and +(45:0.75)] (0,0) -- cycle ;
\end{scope}
\end{scope}
\end{scope}

\begin{scope}[shift={(18,0)}] 
\draw (0,0) -- (0,2) -- (0,2) to [controls=+(45:1.5) and +(135:1.5)] (4,2) -- (4,0) to [controls=+(135:0.75) and +(45:0.75)] (2,0) to [controls=+(135:0.75) and +(45:0.75)] (0,0) -- cycle ; 
\begin{scope}[shift={(4,0)}]
\draw (0,0) -- (0,2) -- (0,2) to [controls=+(45:1.5) and +(135:1.5)] (4,2) -- (4,0) to [controls=+(135:0.75) and +(45:0.75)] (2,0) to [controls=+(135:0.75) and +(45:0.75)] (0,0) -- cycle ;
\end{scope}
\begin{scope}[shift={(8,0)}]
\draw (0,0) -- (0,2) -- (0,2) to [controls=+(45:1.5) and +(135:1.5)] (4,2) -- (4,0) to [controls=+(135:0.75) and +(45:0.75)] (2,0) to [controls=+(135:0.75) and +(45:0.75)] (0,0) -- cycle ;
\end{scope} 
\begin{scope}[shift={(12,0)}]
\draw (0,0) -- (0,2) -- (0,2) to [controls=+(45:1.5) and +(135:1.5)] (4,2) -- (4,0) to [controls=+(135:0.75) and +(45:0.75)] (2,0) to [controls=+(135:0.75) and +(45:0.75)] (0,0) -- cycle ;
\end{scope} 
\end{scope}
 
\end{tikzpicture}

\caption{The supports $\U_0$, $\U_1$, $\U_2 $ and $\L_1$.} 
\label{figure:basic_supports}
\end{figure}
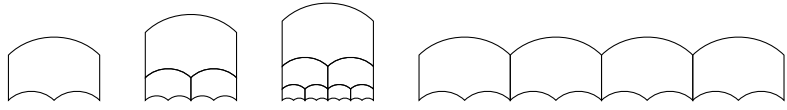

We say that a pattern $p\in A^{\U_n}$ \emph{appears} in a configuration $x\in A^{\H^2}$ if there exists some position $g\in\{ \alpha,\alpha^{-1},\beta,\beta^{-1}\}^*$ such that $p=x_{|g\cdot\U_n}$. Let $F$ be a set of patterns, it defines a \emph{set of colourings} $\Sigma_F\subseteq A^{\H^2}$ as the set of configurations that avoid every pattern in F
$$\Sigma_F=\left\{ x\in A^{\H^2}:\text{ no pattern of }F\text{ appears in x}\right\}.$$

This notion of set of colourings is very close to the classical notion of subshift in symbolic dynamics -- at least from a combinatorial point of view -- but in the case of the $2$-fold horocyclic tessellation, we lack a real shift action to properly define subshifts as dynamical objects.


\begin{definition}
A set of colourings $\Sigma\subseteq A^{\H^2}$ is
\begin{enumerate}
 \item of \emph{finite type} (CFT) if there exists a finite set of forbidden patterns that defines it;
 \item \emph{sofic} if there exists a CFT $\Sigma'\subseteq B^{\H^2}$ and a letter-to-letter map $\Phi:B\rightarrow A$ such that $\Sigma=\Phi(\Sigma')$;
 \item \emph{effective} if there exists a recursively enumerable set of forbidden patterns that defines it.
\end{enumerate}
\end{definition}

\subsection{Computation of Turing machine inside a CFT}

The idea is to embed Turing machine computations the same way it is usually done in $\Z^2$. In this aim, we encode by local rules the lattice $\Z\times\N$ inside $\H^2$ (see the tiles marked by a $\bullet$ symbol in Figure~\ref{figure.MT_hyperbolic}). We denote this set of colourings by $\Sigma_{\Z\times\N}\subset\{\bullet,\emptyset\}^{\H^2}$. Obviously, $\Sigma_{\Z\times\N}$ is CFT, but contains the uniform configuration $\emptyset^{\H^2}$ if we do not impose more constraints.


Once one has this lattice, it is possible to encode the behaviour of any given Turing machine by local rules (see again Figure~\ref{figure.MT_hyperbolic}). Remember that a Turing machine is a model of calculation composed by a finite automaton -- the computation head -- that can be in different states and moves on an infinite tape divided into boxes, each box containing a letter that can be modified by the head.


We do not give in details the alphabet $A_{\mathcal{M}}$ used to encode computations of a given Turing machine $\mathcal{M}$, but all letters in $A_{\mathcal{M}}$ and local rules appear in Figure~\ref{figure.MT_hyperbolic}. The idea is basically to use the encoding $\Sigma_{\Z\times\N}$ of the lattice $\Z\times\N$ and to adapt the classical construction by Wang~\cite{Wang1961} to make sure that we get the space-time diagram of the Turing machine $\mathcal{M}$ starting on the empty word. When we fix a $\bullet$-line in $\H^2$, where the initial state $q_0$ appears once on a blank tape -- it is the same as just fixing an initial seed tile --, one can easily get a computation zone which is infinite both in space and time (see Figure~\ref{figure.MT_hyperbolic}). 

Without a lot of effort, one can deduce from this construction the undecidability of the origin constrained tiling problem in the hyperbolic plane (see~\cite{Robinson1978} for the original proof). The same result without the origin constrained assumption was proved only thirty years later by two different techniques (see~\cite{DBLP:conf/mcu/Kari07} for a proof on the 2-fold horocyclic tessellation of the hyperbolic plane and~\cite{Margenstern2008} for a proof on $(7,3)$-tessellation of Poincar\'e disc with heptagons).

\begin{figure}
\centering
\begin{adjustbox}{
  addcode={\begin{minipage}{\width}}{\caption{We consider the Turing machine $\MT_\texttt{ex}$ that enumerates on its tape the words $ab,aabb,aaabbb,\dots$ and never halts. This machine uses the three letters alphabet $\{a,b,\parallel \}$ and five states $Q=\{ q_0, q_\texttt{a+}, q_\texttt{b+}, q_\texttt{b++},q_\parallel \}$. A separation symbol $\parallel$ is written at the end of each $a^n b^n$. On the top, an example of computation encoded inside a $4\times4$-grid in $\Z^2$. On the bottom, the same computation encoded inside a $4\times4$-grid in the hyperbolic plane: the grid is marked by $\bullet$ symbols.}\label{figure.MT_hyperbolic}\end{minipage}},rotate=90,center}

\end{adjustbox}
\end{figure}

\section{Dyadic encoding in $2$-fold horocyclic tessellations of $\H^2$} 
\label{section:dyadic_encoding} 
 
In this section we present a transformation on subshifts, that preserve both soficness and the property of being of finite type, that is based on the observation that every row of the hyperbolic half-plane contains twice as many cells as the row directly above it.

\subsection{Encoding on a single row all rows above it}

We define a global function on configurations $\Phi$ that doubles the alphabet

$$
\Phi\colon A^{\H^2} \rightarrow \left(A\times A\right)^{\H^2} \cup \mathcal{P}_{|A|}\left(\left(A\times A\right)^{\H^2}\right),
$$
where $\mathcal{P}_{n}(X)$ stands for subsets of size $n$ of $S$, and that will be given by a very simple local rule. Note that in the case where the tessellation chosen contains no fracture line, then the function $\Phi$ is simply defined as  $\Phi\colon A^{\H^2} \rightarrow \left(A\times A\right)^{\H^2}$. If the tessellation chosen contains a fracture line, then there will be one undetermined symbol (hence the cardinality $|A|$ in the definition of $\Phi$) that will propagate along the fracture line, and which can be interpreted as coming from the infinity.

Given a configuration $x\in A^{\H^2}$, we define $\Phi(x)$ as the configuration of $\left(A\times A\right)^{\H^2}$ such that its restriction to the first letter gives $x$, and the restriction to the second letter is given by the local rule pictured on Figure~\ref{figure:local_rule_dyadic_encoding}. Formally, for any position $g$ in $\H^2$, one has

\begin{align*}
\pi_1(\Phi(x)) &= x \\
\pi_2(\Phi(x))_{g\alpha} &= \pi_1(\Phi(x))_{g} \\
\pi_2(\Phi(x))_{g\alpha\beta} &= \pi_2(\Phi(x))_{g}\text{ for every }x\in A^{\H^2}
\end{align*}

where $\pi_i$ denotes the projection on the $i^\textrm{th}$ letter for $i\in\{1,2\}$. Consequently for all $n\geq 1$ one has  $$x_g=\pi_1(\Phi(x))_{g}=\pi_2(\Phi(x))_{g\alpha\left(\alpha\beta\right)^{n-1}}=\pi_2(\Phi(x))_{g\alpha^n\beta^{2^{n-1}-1}},$$
in other terms, to find a symbol of $x$ in $\Phi(x)$ $n$ rows below, one can either go to the row immediately below on the left ($\alpha$) and then $n-1$ times to the row below on the right ($\left(\alpha\beta\right)^{n-1}$), or first go $n$ row below on the left ($\alpha^n$) and then $2^{n-1}-1$ to the right ($\beta^{2^{n-1}-1}$).

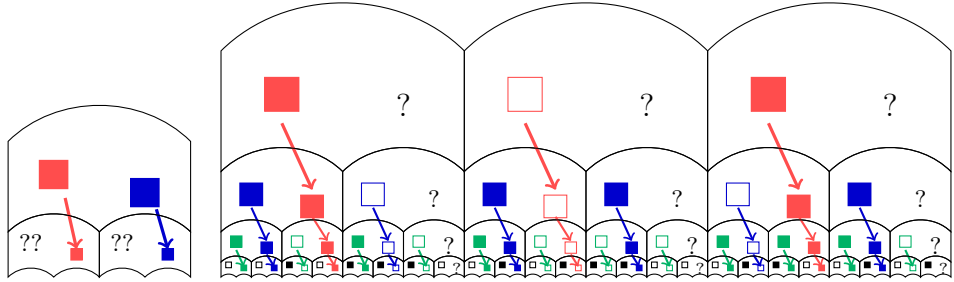
\begin{figure}[H]
\centering
\begin{tikzpicture}[scale=0.2]

\begin{scope}[scale=0.75]
\hypertiledyadic{0}{0}{4}{\bnoirrouge{1.25}}{\bnoirbleu{1.25}}
\draw[very thick,color=rouge,->] (5,3) -> (6,-1.5);
\hypertiledyadic{0}{-4}{2}{??}{\bnoirrouge{0.5}}
\draw[very thick,color=bleu,->] (13,2) -> (14,-1.5);
\hypertiledyadic{8}{-4}{2}{??}{\bnoirbleu{0.5}}
\end{scope}

\begin{scope}[shift={(14,0)}] 
\hypertiledyadic{0}{4}{4}{\bnoirrouge{1.5}}{\large{?}}
\fleche{0}{4}{4}{rouge}{very thick}
\hypertiledyadic{16}{4}{4}{\bblancrouge{1.5}}{\large{?}}
\fleche{16}{4}{4}{rouge}{very thick}
\hypertiledyadic{32}{4}{4}{\bnoirrouge{1.5}}{\large{?}}
\fleche{32}{4}{4}{rouge}{very thick}

\hypertiledyadic{0}{0}{2}{\bnoirbleu{1}}{\bnoirrouge{1}}
\flechebis{0}{0}{2}{rouge}{thick}
\hypertiledyadic{8}{0}{2}{\bblancbleu{1}}{?}
\hypertiledyadic{16}{0}{2}{\bnoirbleu{1}}{\bblancrouge{1}}
\flechebis{16}{0}{2}{rouge}{thick}
\hypertiledyadic{24}{0}{2}{\bnoirbleu{1}}{?}
\hypertiledyadic{32}{0}{2}{\bblancbleu{1}}{\bnoirrouge{1}}
\flechebis{32}{0}{2}{rouge}{thick}
\hypertiledyadic{40}{0}{2}{\bnoirbleu{1}}{?}

\fleche{0}{0}{2}{bleu}{thick}
\fleche{8}{0}{2}{bleu}{thick}
\fleche{16}{0}{2}{bleu}{thick}
\fleche{24}{0}{2}{bleu}{thick}
\fleche{32}{0}{2}{bleu}{thick}
\fleche{40}{0}{2}{bleu}{thick}

\hypertiledyadic{0}{-2}{1}{\bnoirvert{0.5}}{\bnoirbleu{0.5}}
\flechebis{0}{-2}{1}{bleu}{thick}
\hypertiledyadic{4}{-2}{1}{\bblancvert{0.5}}{\bnoirrouge{0.5}}
\flechebis{4}{-2}{1}{rouge}{thick}
\hypertiledyadic{8}{-2}{1}{\bnoirvert{0.5}}{\bblancbleu{0.5}}
\flechebis{8}{-2}{1}{bleu}{thick}
\hypertiledyadic{12}{-2}{1}{\bblancvert{0.5}}{\footnotesize{?}}
\hypertiledyadic{16}{-2}{1}{\bnoirvert{0.5}}{\bnoirbleu{0.5}}
\flechebis{16}{-2}{1}{bleu}{thick}
\hypertiledyadic{20}{-2}{1}{\bblancvert{0.5}}{\bblancrouge{0.5}}
\flechebis{20}{-2}{1}{rouge}{thick}
\hypertiledyadic{24}{-2}{1}{\bblancvert{0.5}}{\bnoirbleu{0.5}}
\flechebis{24}{-2}{1}{bleu}{thick}
\hypertiledyadic{28}{-2}{1}{\bblancvert{0.5}}{\footnotesize{?}}
\hypertiledyadic{32}{-2}{1}{\bnoirvert{0.5}}{\bblancbleu{0.5}}
\flechebis{32}{-2}{1}{bleu}{thick}
\hypertiledyadic{36}{-2}{1}{\bnoirvert{0.5}}{\bnoirrouge{0.5}}
\flechebis{36}{-2}{1}{rouge}{thick}
\hypertiledyadic{40}{-2}{1}{\bnoirvert{0.5}}{\bnoirbleu{0.5}}
\flechebis{40}{-2}{1}{bleu}{thick}
\hypertiledyadic{44}{-2}{1}{\bblancvert{0.5}}{\footnotesize{?}}

\fleche{0}{-2}{1}{vert}{thick}
\fleche{4}{-2}{1}{vert}{thick}
\fleche{8}{-2}{1}{vert}{thick}
\fleche{12}{-2}{1}{vert}{thick}
\fleche{16}{-2}{1}{vert}{thick}
\fleche{20}{-2}{1}{vert}{thick}
\fleche{24}{-2}{1}{vert}{thick}
\fleche{28}{-2}{1}{vert}{thick}
\fleche{32}{-2}{1}{vert}{thick}
\fleche{36}{-2}{1}{vert}{thick}
\fleche{40}{-2}{1}{vert}{thick}
\fleche{44}{-2}{1}{vert}{thick}

\hypertiledyadic{0}{-3}{0.5}{\bblanc{0.25}}{\bnoirvert{0.25}}
\hypertiledyadic{2}{-3}{0.5}{\bblanc{0.25}}{\bnoirbleu{0.25}}
\hypertiledyadic{4}{-3}{0.5}{\bnoir{0.25}}{\bblancvert{0.25}}
\hypertiledyadic{6}{-3}{0.5}{\bblanc{0.25}}{\bnoirrouge{0.25}}
\hypertiledyadic{8}{-3}{0.5}{\bnoir{0.25}}{\bnoirvert{0.25}}
\hypertiledyadic{10}{-3}{0.5}{\bnoir{0.25}}{\bblancbleu{0.25}}
\hypertiledyadic{12}{-3}{0.5}{\bnoir{0.25}}{\bblancvert{0.25}}
\hypertiledyadic{14}{-3}{0.5}{\bblanc{0.25}}{\tiny{?}}
\hypertiledyadic{16}{-3}{0.5}{\bblanc{0.25}}{\bnoirvert{0.25}}
\hypertiledyadic{18}{-3}{0.5}{\bnoir{0.25}}{\bnoirbleu{0.25}}
\hypertiledyadic{20}{-3}{0.5}{\bblanc{0.25}}{\bblancvert{0.25}}
\hypertiledyadic{22}{-3}{0.5}{\bnoir{0.25}}{\bblancrouge{0.25}}
\hypertiledyadic{24}{-3}{0.5}{\bnoir{0.25}}{\bblancvert{0.25}}
\hypertiledyadic{26}{-3}{0.5}{\bnoir{0.25}}{\bnoirbleu{0.25}}
\hypertiledyadic{28}{-3}{0.5}{\bblanc{0.25}}{\bblancvert{0.25}}
\hypertiledyadic{30}{-3}{0.5}{\bblanc{0.25}}{\tiny{?}}
\hypertiledyadic{32}{-3}{0.5}{\bblanc{0.25}}{\bnoirvert{0.25}}
\hypertiledyadic{34}{-3}{0.5}{\bnoir{0.25}}{\bblancbleu{0.25}}
\hypertiledyadic{36}{-3}{0.5}{\bnoir{0.25}}{\bnoirvert{0.25}}
\hypertiledyadic{38}{-3}{0.5}{\bblanc{0.25}}{\bnoirrouge{0.25}}
\hypertiledyadic{40}{-3}{0.5}{\bblanc{0.25}}{\bnoirvert{0.25}}
\hypertiledyadic{42}{-3}{0.5}{\bnoir{0.25}}{\bnoirbleu{0.25}}
\hypertiledyadic{44}{-3}{0.5}{\bblanc{0.25}}{\bblancvert{0.25}}
\hypertiledyadic{46}{-3}{0.5}{\bnoir{0.25}}{\tiny{?}}

\end{scope}

\end{tikzpicture}
\caption{On the left: the local rules from which one can deduce the configuration $\Phi(x)$ from the configuration $x$. On the right: an example of configuration $\Phi(x)$ on the alphabet $\{\blacksquare,\square\}$ -- coloured version to better visualize the local rule.}
\label{figure:local_rule_dyadic_encoding}
\end{figure}

For a given subshift $\Sigma\subseteq A^{\H^2}$, the subshift $\Phi(\Sigma)\subseteq \left(A\times A\right)^{\H^2}$ is called the \define{dyadic encoding of $\Sigma$}. Obviously if $\Sigma$ is SFT (resp. sofic, effective), then so is $\Phi(\Sigma)$.

\subsection{Detecting patterns}
\label{subsection.detecting_patterns}

In this section we describe more precisely how the transformation $\Phi$ acts on allowed and forbidden patterns of a set of colourings $\Sigma$.
 
Let $p\in A^{\U_n}$, define $\widetilde{p}=\{\pi_2(\Phi(x))_{\L_n} :x\in A^{\H^2}\text{ such that }x_{\U_n}=p\}\subset A^{2^n}$ the set of linear patterns of size $n$ which appear in the bottom of the pattern $p$ after application of $\Phi$. Some letters of an element in $\widetilde{p}$ code letters in $p$, and others code letters that appears in $x$ outside $\U_n$. If there is no ambiguity -- i.e. if it codes a letter in $p$ -- denote $\widetilde{p}_i$ as a letter.
 
\begin{proposition}\label{proposition:ValueBottom}
The pattern $p\in A^{\U_n}$ appears in a configuration $x\in A^{\H^2}$ in position $g$ (i.e. $p=x_{g\cdot\U_n}$) if and only if  an element of $\widetilde{p}$ appears in $\Phi(x)$ in position $g\cdot\L_n$ (i.e. $\pi_2(\Phi(x))_{g\cdot\L_n}\in\widetilde{p}$).
\end{proposition}


Proposition~\ref{proposition:ValueBottom} means that the whole information about a pattern with support $\U_n$ is entirely contained in a linear pattern with support $\L_n$ of its dyadic encoding. Thus looking for occurrences of a pattern $p$ in a configuration $x$ is the same as looking for occurrences of $\widetilde{p}$ in the configuration~$\Phi(x)$.
 
\begin{proposition}\label{proposition:2}
Let $x\in A^{\H^2}$, $g\in\H^2$ and $n\in\N$. Consider the pattern $p=x_{g\U_n}$ of support $\U_n$, then
$$\pi_2(\Phi(x))_{g\alpha^{n+k}\beta^{2^i+2^k-1}}=\widetilde{p}_i\textrm{ for all }k\geq 1 \textrm{ and } i\in\{0,\dots|\widetilde{p}|-1\}.$$
\end{proposition}

\begin{proof}
From Proposition~\ref{proposition:ValueBottom} we deduce that $\pi_2(\Phi(x))_{g\alpha^{n}\beta^{2^{i-1}-1}}=\widetilde{p}_i$ for all $i\in\{0,\dots|\widetilde{p}|-1\}$. The result follows from the fact that for all $k\geq 1$ one has $\pi_2(\Phi(x))_{g\left(\alpha\beta\right)^k}=\pi_2(\Phi(x))_{g}$ and $\beta^i\alpha^k=\alpha^k\beta^{2^i}$. 
\end{proof}

From Proposition~\ref{proposition:ValueBottom} we deduce that given a pattern $p\in A^{\U_n}$ and a positive integer $k$, it is algorithmically possible to generate the dyadic encodings of patterns $\widetilde{p}^{(k)}$ located $k$ rows below $p$ (and thus $\widetilde{p}=\widetilde{p}^{(1)}$): the way linear patterns $\widetilde{p}^{(k)}$ are spit can be encoded inside a Turing machine. Define $\textrm{split}(p)=\cup_{k\in\N^*}\widetilde{p}^{(k)}$, we deduce the following.

\begin{proposition}\label{proposition:3}
Let $x\in A^{\H^2}$. Then $p$ appears in $x$ if and only if for every $y\in\Phi(x)$, there exists $p'\in\textrm{split}(p)$ such that $p'$ appears in $y$.
\end{proposition}

\section{Effective sets of colourings are sofic on the hyperbolic half-plane} 
\label{section:result}  
 
We say that a set of colourings $\mathbf{T}\subset A^{\H^2}$ is \emph{row-constrained} if there exists a special symbol $\approx\in A$ that appears in every configuration $x\in\mathbf{T}$, and such that its presence forces all letters on rows above it (including letters on the same row) to be also $\approx$ -- we say that the letter $\approx$ has the \emph{half-plane property}. In this section we prove the following result. 
 
\begin{theorem}
 Any row-constrained effective set of colourings on $\H^2$ is sofic. In other words, if $\mathbf{\T}$ is an effective set of colourings on $A$ and if one letter $\approx\in A$ has the half-plane property, then the row-constrained set of colourings $\T\cap\left\{ x\in A^{\H^2}, \approx\text{ appears in }x \right\}$ is sofic. 
\end{theorem}

For more readability, if $x$ is a configuration of a row-constrained set of colourings, we do not picture the half-plane filled with $\approx$ but letters below this half plane with a double line on the top of the pentagon (see Figure~\ref{figure:unconstrained_computation_zones} for instance). 
 
\subsection{Sketch of the proof}

We will encode Turing machine computations inside a row-constrained CFT on $\H^2$. Thanks to the dyadic encoding presented in Section~\ref{section:dyadic_encoding} it is enough to check the occurrences of forbidden patterns produced by this machine on infinitely many rows in the dyadic encoding of the original configuration.

\subsection{A four layers construction}
\label{subsection:four_layers}

Let $\mathbf{T}$ be a row-constrained effective set of colourings on some alphabet~$A$. Let $\mathcal{M}$ be a Turing machine that enumerates a set of forbidden patterns for $\mathbf{T}$ -- we assume that the machine runs on a one-sided tape. We construct a four layers row-constrained CFT $\Sigma_{\mathcal{M}}=\Sigma_1\times\Sigma_2\times\Sigma_3\times\Sigma_4$. 

\paragraph{\textbf{First layer: configurations in $A^{\H^2}$.}} The first layer $\Sigma_1$ only contains configurations in $A^{\H^2}$, with no constraint on them, except the ones that will be given by interaction with other layers.

\paragraph{\textbf{Second layer: computation zones.}} The second layer $\Sigma_2$ contains computation zones. First define the row-constrained CFT $\Sigma_2'$ on alphabet $\left\{ a,b,a^*,b^*,c\right\}$ as the one defined by the set of allowed patterns with support $\U_1\cup\beta\cdot\U_1$ appearing in Figure~\ref{figure:unconstrained_computation_zones}. By local interaction with the first layer, we force that if a product letter $(\approx,\lambda)$ appears, then $\lambda=c$ and that in $\Sigma_2$, the row immediately below the lowest row of $c$ is $^\infty(a^*b^*)^\infty$. Then $\Sigma_2$ is the row-constrained CFT on the product alphabet $\left\{ a,b,a^*,b^*,c\right\}\times\{\bullet,\emptyset\}$, seen as a subset of $\Sigma_2'\times\Sigma_{\Z\times\N}$ with the additional rule that a letter of the product alphabet with either $a^*$ or $b^*$ on its first coordinate is always associated with $\bullet$ on its second coordinate.

\begin{figure}[ht!]
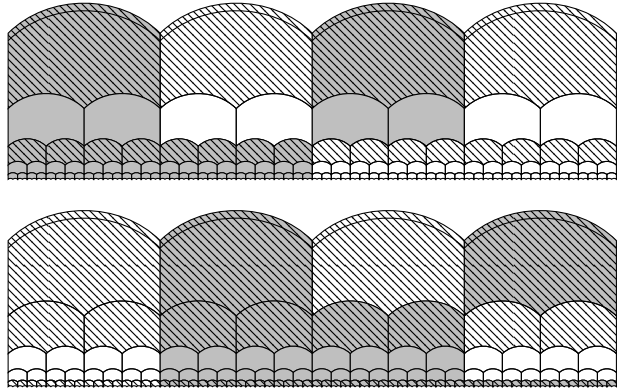

\centering


\caption{Two examples of row-constrained configurations on $\left\{ a,b,a^*,b^*,c\right\}$ in $\Sigma_2'$. For more readability, all letters $c$ are omitted and replaced by a double top edge, and letters $a$ and $a^*$ are represented by Grey cells, letters $b$ and $b^*$ by white cells and letters $a^*$ and $b^*$ with stripes. The first row is always a $^\infty(a^*b^*)^\infty$ row, and $*$-rows and non-$*$-rows do not necessarily alternate. A tile below a $*$-tile is of the same type ($a$ for $a^*$ and $b$ for $b^*$). Type $a$ and type $b$ zones always merge on a $*$-row.}
\label{figure:unconstrained_computation_zones}
\end{figure}

Let $x\in\Sigma_2$. We call \emph{type $a$ computation zone} (resp. \emph{type $b$ computation zone}) a pattern uniformly filled with product letters with $a$'s and $a^*$'s (resp. $b$'s and $b^*$'s) on its first coordinate, with support $\bigcup_{i=1}^{k} \beta^i\cdot\U_n$ where $n$ -- the \emph{height} -- and $k$ -- the \emph{width} --  are maximal, that appears in $x$. This second layer is made such that every configuration in $\Sigma_2$ is made of computation zones of type $a$ and $b$ that alternate, and type $a$ zones merge with their type $b$ right neighbour to form a larger computation zone. With no more constraints, these patterns may define infinitely high computation zones of bounded width at some point, as suggested in Figure~\ref{figure:unconstrained_computation_zones} -- all computation zones are of width $2^k$ thank to the initial alternation of $a^*$ and $b^*$. We will see how this problem is fixed by interacting with the 
fourth layer: the Turing machine itself will choose the height of every computation zone so that it can perform as many steps of calculation as the width of the computation zone allows.

\paragraph{\textbf{Third layer: dyadic encoding.}} The third layer $\Sigma_3$ contains one (since there can be an undetermined symbol) dyadic encoding of the first layer $\Phi(\Sigma_1)$. This can be done using local rules between first and third layers.

\paragraph{\textbf{Fourth layer: Turing machine calculations.}} The fourth layer $\Sigma_4$ contains calculations of a Turing machine $\widetilde{\MT}$ with the following behaviour
\begin{enumerate}
 \item the machine $\widetilde{\MT}$ has two tapes, the \emph{computation tape} that is initially filled with blank symbols $\sharp$, whose width coincides with the width of the computation zone, and the \emph{detecting tape} whose width coincides with the width of the computation zone and its left and right neighbours (hence three times wider than the computation tape) on which the first row of $\Sigma_1$ is copied out. The widths and overlaps of detecting tapes ensure that any linear pattern will eventually be contained in a single detecting tape. Note that at most three detecting tapes may overlap, which ensures finiteness of the alphabet needed.
 \item the machine $\widetilde{\MT}$ simulates $\MT$, and for each pattern produced by $\MT$ transforms it into the set $\mathcal{P}$ of patterns in $\textrm{split}(p)$ -- defined in Section~\ref{subsection.detecting_patterns} -- that can fit inside the computation zone (thus $\mathcal{P}$ is always finite)
 \item it checks whether elements of $\mathcal{P}$ appear or not on the detecting tape
 \item if a pattern of $\mathcal{P}$ is detected, then the machine $\widetilde{\MT}$ instantaneously reaches a special state $q_f$, that will be forbidden in the final set of colourings
 \item once the head of calculation tries to go to the right of the rightmost cell in the current computation zone, the computation zone is closed and merges with its right or left neighbour to get a twice bigger computation zone. This ensures that if the Turing machine lacks some space on its tape in some computation zone, there exists a bigger zone in which the machine can perform more steps of calculation (it may actually happen that all computation zones have bounded width; this is the case when $\Sigma$ is CFT). All this can be done using local rules between second and fourth layers.
\end{enumerate}
All conditions above ensure that a computation zone located at $g\cdot\bigcup_{i=1}^{k} \beta^i\cdot\U_n$ checks whether a pattern from $\widetilde{p}$ appears in $g.\L_n$, i.e. whether the pattern $p$ appears somewhere above the computation zone.

\subsection{Main result}
 

\begin{customthm}{1}
 Row-constrained effective sets of colourings on $\H^2$ are sofic.
\end{customthm}

\begin{proof}
Let $\T$ be a row-constrained effective set of colourings. We denote by $\Sigma=\pi_1\left(\Sigma_{\mathcal{M}}\right)$ where $\Sigma_{\mathcal{M}}$ is the sofic set of colourings described in Section~\ref{subsection:four_layers}. It is straightforward that the construction described in Section~\ref{subsection:four_layers} provides a sofic set of colourings $\Sigma$ such that $\T\subseteq\Sigma$. It remains to prove that the local rules defining $\Sigma$ forces all configurations $x\in\Sigma$ to also belong to $\T$.  

Let $x$ be in $\Sigma$, and suppose that a pattern $p=\pi_1(x_{g\cdot\U_n})$ is forbidden in $\T$. Then the Turing machine that enumerates forbidden patterns in $\T$ will enumerate $p$ at some point. By construction of the computation zones on the second layer, there exists a computation zone large enough to check whether one element of $\textrm{split}(p)$ appears on its associated detecting tape (and one can assume it is located below $p$, if not take a bigger zone below it). By Proposition~\ref{proposition:3}, it is enough to ensure that any forbidden pattern will eventually be detected, consequently $x$ is in $\T$.
\end{proof}

\section{Conclusion}

The ideas presented in this article constitute only a first step in proving Conjecture~\ref{conjecture}. The natural idea would consist in using Goodman-Strauss hierarchical aperiodic tiling of the hyperbolic plane~\cite{DBLP:journals/tcs/Goodman-Strauss10} to define arbitrarily large computation zones without any constraints. But a few points still remain to be clarified. 

\bibliographystyle{alpha}
\bibliography{effectifs_demi_plan_hyperbolique}

\end{document}